\documentclass[aps,nopacs,nokeys,superscriptaddress,11pt,a4paper,twoside,notitlepage]{revtex4}

\usepackage{amssymb}	
\usepackage{graphicx} 
\usepackage{color}
\usepackage[dvips]{geometry}
\usepackage{amsmath,enumerate,anysize,amsfonts,amsthm,mathrsfs,latexsym}

\newcommand{\ket}[1]{\left| #1 \right>} 
\newcommand{\bra}[1]{\left< #1 \right|} 
\let\baraccent=\= 
\renewcommand{\=}[1]{\stackrel{#1}{=}} 
\newcommand{\trace}{\mathrm{Tr}}
\newcommand{\boldvec}[1]{\underline{\bf#1}}
\newcommand{\hilb}[1]{\mathscr{H}_{#1}}
\newcommand{\real}[0]{\mathbb{R}}

\newcommand{\powset}[1]{\mathbb{P}(#1)}

\newtheorem{prop}{Proposition}
\newtheorem{thm}[prop]{Theorem}
\newtheorem*{thm1'}{Theorem 1$'$}
\newtheorem*{thm2'}{Theorem 5$'$}

\theoremstyle{definition}

\theoremstyle{remark}

\marginsize{2.5cm}{2.5cm}{2cm}{1.5cm}

\newcommand{\dotcup}{\ensuremath{\mathaccent\cdot\cup}}

\renewcommand{\emptyset}[0]{\text{\O}}

\begin{document}

\title{Infinitely many constrained inequalities for the von Neumann entropy}

\author{Josh Cadney}
\email{josh.cadney@bristol.ac.uk}
\affiliation{Department of Mathematics, University of Bristol, Bristol BS8 1TW, U.K.}

\author{Noah Linden}
\affiliation{Department of Mathematics, University of Bristol, Bristol BS8 1TW, U.K.}

\author{Andreas Winter}
\affiliation{Department of Mathematics, University of Bristol, Bristol BS8 1TW, U.K.}
\affiliation{Centre for Quantum Technologies, National University of Singapore, Singapore 117542}

\date{23 June 2011}

\begin{abstract}
  We exhibit infinitely many new, constrained inequalities for the von Neumann
  entropy, and show that they are independent of each other and the known
  inequalities obeyed by the von Neumann entropy (basically strong subadditivity).
  The new inequalities were proved originally by Makarychev \emph{et al.}
  [Commun. Inf. Syst., 2(2):147-166, 2002]
  for the Shannon entropy, using properties of probability distributions.
  Our approach extends the proof of the inequalities to the quantum domain,
  and includes their independence for the quantum and also the classical cases.
\end{abstract}

\maketitle

\section{Introduction}
The von Neumann entropy, given by $S(\rho)=-\trace{\rho\log\rho}$ for a quantum 
state (density operator) $\rho$, is one of the cornerstones of quantum information
theory. It plays an essential role in the expressions for the best achievable rates
of virtually every coding theorem. In particular, when proving the optimality of these
expressions, it is the inequalities governing the relative magnitudes of the
entropies of different subsystems which are important. There are essentially two
such inequalities known, the so called basic inequalities:
\begin{align}
I(A:B|C) := -S(C)+S(AC)+S(BC)-S(ABC) &\geq 0, \tag{SSA}\label{SSA}\\
S(AB)+S(AC)-S(B)-S(C)    &\geq 0. \tag{WMO}\label{weakmono}
\end{align}
Inequality \eqref{SSA} is known as \emph{strong subadditivity}
and was proved by Lieb and Ruskai \cite{LR73} and the expression on the
left hand side as the (quantum) conditional mutual information;
inequality \eqref{weakmono} is usually called \emph{weak monotonicity},
and it is in fact equivalent to (SSA) -- see section \ref{notation} below.

To be precise, we will be considering only \emph{linear} inequalities involving
the entropies of various reduced states of a multi-party quantum state, as we
shall explain, and partly motivate now.
Given a multipartite state $\rho$ on a set of parties (quantum systems) $N=\{X_1,\ldots,X_n\}$,
we can think of the entropy as a function which assigns a real number to each subset of $N$,
i.e. $S(.)_\rho:\powset{N}\rightarrow\real$ with $S(J)_\rho:=S(\rho_J)$.
(We will use the notation $S(J)_\rho$ and $S(\rho_J)$ interchangeably).
Further, with each function $f:\powset{N}\rightarrow\real$, which satisfies $f(\emptyset)=0$,
we can associate a vector in $2^n-1$ dimensional real space: $(f(X_1),f(X_2),\ldots,f(X_1\ldots X_n))$.
It is then natural to ask the question: which vectors can arise as the entropies of quantum states?
For example, the vector $(1,1,2)$ is the entropy vector of the maximally mixed state on
two spin-$\frac12$ systems. However, we know that the vector $(1,1,3)$ cannot represent
the entropies of any quantum state since in general the quantity
$S(X_1)_\rho+S(X_2)_\rho-S(X_1X_2)_\rho$ is non-negative,
whereas here it is equal to $-1$.
Thus, the question of which vectors can be realised by quantum states
is inextricably linked to the knowledge of entropy inequalities. Indeed, for $n=2$ and $n=3$
it has been shown \cite{Pip03} that the closure of the set of achievable vectors, which we
will denote $\overline{\Sigma}_n^*$, is exactly the cone in $\real^{2^n-1}$ cut out by the
basic inequalities, denoted $\Sigma_n$. In other words, a vector can be realised, with
arbitrary accuracy, as the entropy vector of a quantum state if and only if it satisfies
all the basic inequalities. For $n\geq 4$ it can again be shown that
$\overline{\Sigma}_n^*$ forms a convex cone, however, it is unknown whether or not
this cone is the same as that which is determined by the basic inequalities.

In classical information theory, the Shannon entropy of a random variable, given by
$H(X)=-\sum_{x\in\mathcal{X}}p_X(x)\log p_X(x)$, plays an analagous role to the von
Neumann entropy. It satisfies the same basic inequalities as above, with \eqref{weakmono}
replaced by the stronger condition of \emph{monotonicity}: $H(AB)\geq H(A)$.
The analogous classical problem to the one we study here has been extensively studied for quite some time.
First Zhang and Yeung \cite{YZ98}, and then Makarychev \emph{et al.} \cite{Mak02}
and Dougherty \emph{et al.} \cite{Dou06} found new inequalities, which are not implied by
the basic inequalities. Mat\'{u}\v{s} \cite{Mat07} even proved that for $n\geq 4$ the
classical entropy cone is not polyhedral, i.e.~it cannot be described by any finite
set of linear inequalities.

In the quantum case, only one inequality is known which cannot be deduced from the basic inequalities \cite{LW05}, and it is a so-called constrained inequality -- an inequality which holds whenever certain conditional mutual informations are zero. This shows that parts of certain faces of the cone $\Sigma_n$ do not contain any entropy vectors of quantum states (noting that $\Sigma_n$, being defined by finitely
many linear inequalities, is a polyhedral cone).
This is not enough, however, to conclude that the entropy cone, $\overline{\Sigma}_n^*$, is strictly smaller, as we are concerned with the closure. In fact, it remains a major open problem to decide
the existence of an unconstrained inequality for the von Neumann entropy that is not
implied by the basic inequalities.

Here we make progress in two different directions; we prove an infinite family of constrained inequalities, which are provably independent, and we do so with a strictly smaller set of constraints.
The structure of the remainder of the paper is as follows: in section \ref{notation} we introduce
notation and review the basic linear framework of entropy inequalities.
In section \ref{result} we prove that a family of constrained inequalities are true for
the von Neumann entropy of quantum states; in section \ref{independence} we show
that this family is mutually independent; in section \ref{variations} we exhibit some
alternate forms of the inequalities; and in section \ref{discussion} we conclude and
mention some open problems.

\section{Entropy Inequalities} \label{notation}
In this section we explain preciesely what it means for entropy inequalities to hold, and
to be independent of one another. Consider the inequality of strong subadditivity \eqref{SSA},
\begin{equation}
  -S(C)+S(AC)+S(BC)-S(ABC) \geq 0,
\end{equation}
which holds for all quantum states on a Hilbert space $\hilb{A}\otimes\hilb{B}\otimes\hilb{C}$.
By swapping the labels of $\hilb{A}$ and $\hilb{C}$ we obtain another inequality:
$-S(A)+S(AC)+S(AB)-S(ABC)\geq0$. Alternatively, suppose we have a four party quantum state on
$\{A,B,C,D\}$. Then we can think of $\rho_{ABCD}$ as a tripartite state on
$\hilb{A}\otimes\hilb{BC}\otimes\hilb{D}$ and so we obtain the inequality
$-S(D)+S(AD)+S(BCD)-S(ABCD)\geq0$. We could even think of a bipartite state $\rho_{AB}$
as a tripartite state on $\{A,B,C\}$ with $\hilb{C}$ trivial; in this case we obtain
\begin{equation}
  I(A:B) := S(A)+S(B)-S(AB) \geq 0, \label{mutualinfo}
\end{equation}
which sometimes is considered another basic inequality because of the
importance of the expression on the left, the (quantum) mutual information,
although the above reasoning shows that it is not really necessary.
Generally, if we have a state on a set of parties $N$, then for each disjoint triple
$\alpha,\beta,\gamma\subseteq N$ we obtain a different \emph{instance} of SSA.
A function which satisfies all instances of \eqref{SSA} is called \emph{submodular}
\cite{Oxley}.

Weak monotonicity has similarly as special instances
\begin{align}
  S(AB)+S(A)-S(B) &\geq 0, \label{triangle}\\
  S(A)            &\geq 0, \label{pos}
\end{align}
the first known as \emph{triangle inequality}; they are obtained from
\eqref{weakmono} by making $\hilb{C}$ trivial, and both $\hilb{B}$ and $\hilb{C}$
trivial, respectively.

More generally, an entropy inequality is an expression of the form
\begin{equation}\label{ineqdef}
  L(X_1,\ldots,X_k) = \sum_{\omega\in\powset{K}}\chi_\omega S(X_\omega) \geq 0,
\end{equation}
for some $k\in \mathbb{N}$ where $K=\{1,\ldots,k\}$, $\chi_\omega\in\real$ and $X_\omega=\bigcup_{i\in\omega}X_i$. An \emph{instance} of the inequality is the expression obtained by fixing a ground set of parties, $N$, and substituting $X_1,\ldots,X_k$ for $k$ disjoint subsets of $N$ in \eqref{ineqdef}. We then say that the inequality $\tilde{L}\geq 0$ is \emph{independent} of the inequalities $L_1,\ldots,L_m \geq 0$
whenever some instance of $\tilde{L}$ cannot be written as a positive linear combination of instances of $L_1,\ldots,L_m$.

In section \ref{independence} we will prove that a family of \emph{constrained} inequalities
are independent of each other, and of the basic inequalities. A constrained inequality can be thought of in the same way as above, but it is required to hold only when the constraints, $C_i$, are equal to zero:
\begin{equation}
  C_i(X_1,\ldots,X_k) = \sum_{\omega\in\powset{K}} \eta_\omega^{(i)} S(X_\omega) = 0,
\end{equation}
for all $i$. We say that the inequality $\tilde{L}\geq0$ (with constraints $\{C_i\}$) is independent of the inequalities $L_1,\ldots,L_m\geq0$ (each with some subset of $\{C_i\}$ as constraints) whenever some instance of $\tilde{L}$ cannot be written as a positive
linear combination of instances of $L_1,\ldots,L_m$ plus an arbitrary linear combination of the $C_i$. (Here we only consider instances of $L_1,\ldots,L_m$ with constraints matching the particular instance of $\tilde{L}$.)

A slight caveat to this definition of independence is provided by the following observation.
Notice that one instance of \eqref{SSA} applied to a purification $\psi_{ABCD}$ of the state
$\rho_{ABC}$ is
\begin{equation}
  -S(B)+S(AB)+S(BD)-S(ABD) \geq 0.
\end{equation}
Using the property of pure states that $S(J)_\psi=S(J^c)_\psi$, where $J^c = N\setminus J$
is the complement of $J$ in $N$, we can eliminate $D$ from this inequality to obtain
\begin{equation}
  -S(B)-S(C)+S(AB)+S(AC) \geq 0,
\end{equation}
and so we have deduced \eqref{weakmono}.
However, in the sense defined above, the inequalities \eqref{SSA} and \eqref{weakmono}
are independent.

\section{Main Result}\label{result}
Our main result is the following theorem,
whose analogue was proved in \cite{Mak02} for the Shannon entropy:
\begin{thm} \label{main}
Let $\rho$ be a multipartite quantum state on parties
$\{A,B,C,X_1,\ldots,X_n\}$ which satisfies the constraints:
\begin{equation}
  I(A:C|B)_{\rho}=I(B:C|A)_{\rho}=0.
\end{equation}
Then the following inequality holds:
\begin{equation}
  S(X_1\ldots X_n)_\rho+(n-1)I(AB:C)_\rho \leq \sum_{i=1}^n S(X_i)_\rho+\sum_{i=1}^n I(A:B|X_i)_\rho.
\end{equation}
\end{thm}

Before commencing the proof, we state a result from \cite{HJPW04} which will be crucial.
\begin{prop}
A state $\rho_{ABC}$ on $\hilb{A}\otimes\hilb{B}\otimes\hilb{C}$ satisifies $I(A:C|B)_\rho=0$ if and only if there is a decomposition of system $B$ as
\begin{equation}
\hilb{B}=\bigoplus_{j}\hilb{b_j^L}\otimes\hilb{b_j^R}
\end{equation}
into a direct sum of tensor products, such that
\begin{equation}
\rho_{ABC}=\bigoplus_j q_j\rho_{Ab_j^L}^{(j)}\otimes\rho_{b_j^RC}^{(j)},
\end{equation}
with states $\rho_{Ab_j^L}^{(j)}$ on $\hilb{A}\otimes\hilb{b_j^L}$ and
$\rho_{b_j^RC}^{(j)}$ on $\hilb{b_j^R}\otimes\hilb{C}$,
and a probability distribution $\{q_j\}$.
\end{prop}

In \cite{LW05}, using this result, two of the present authors derived the general structure
of a state $\rho_{ABC}$ which saturates two separate instances of \eqref{SSA}
simultaneously, exactly the constraints of Theorem \ref{main}.
They found that such a state must have the form
\begin{equation}
\rho_{ABC}=\bigotimes_{i,j} p_{ij} \sigma_{a_i^L}^{(i)}\otimes\sigma_{a_i^Rb_j^L}^{(ij)}
                                     \otimes\sigma_{b_j^R}^{(j)}\otimes\sigma_C^{(k)}
\end{equation}
where, importantly, $k$ is a function only of $i$ and only of $j$, in the sense that
\begin{equation}
  k = k(i,j) = k_1(i) = k_2(j) \quad\text{ whenever }\quad p_{ij} > 0.
\end{equation}
In particular, $k$ need only be only defined where $p_{ij}>0$ so that it is
not necessarily constant. By collecting the terms of equivalent $k$ we can write
\begin{equation} \label{structure}
\rho_{ABC}=\bigoplus_k p_k \sigma_{AB}^{(k)}\otimes\sigma_C^{(k)},
\end{equation}
where
$p_k\sigma_{AB}^{(k)} = \sum_{i,j:k(i,j)=k}p_{ij}\sigma_{a_i^L}^{(i)}\otimes\sigma_{a_i^Rb_j^L}^{(ij)}
                                                    \otimes\sigma_{b_j^R}^{(j)}$.
We are now ready to proceed with the proof of the main theorem.

\begin{proof}[Proof of Theorem 1]
Since $I(A:C|B)_\rho=I(B:C|A)_\rho=0$, from the argument above we know that $\rho_{ABC}$ has the form \eqref{structure}.
The key ideas of the proof are the following three steps:
\begin{enumerate}[(1)]
\item ``Measure the value of $k$'' locally at $A$, without disturbing the state $\rho_{ABC}$, storing the result of this measurement in a classical register, $R$.
\item Observe that the entropies of our new $(n+4)$-party state, $\sigma$, satisfy many desirable properties, which allow us to derive new inequalities for $\sigma$ by methods analagous to \cite{Mak02}.
\item Relate these inequalities for $\sigma$ back to inequalities for $\rho$, using the fact that the measurement left many of the entropies unchanged.
\end{enumerate}

More precisely, we define $\hilb{A}^{(k)}=\bigoplus_{j:k_1(j)=k}\hilb{a_j^L}\otimes\hilb{a_j^R}$ and $P_k$ to be the projection operator onto $\hilb{A}^{(k)}$. We then perform the local projective measurement at A defined by the projectors $\{P_k\}$ and, conditional upon obtaining measurement outcome $i$, we prepare the state $\ket{i}\bra{i}$ in an ancilla, $R$, where $\{\ket{i}\}$ form an orthonormal basis of $\hilb{R}$. We then forget the value of $i$.

Let $p_k$ be the probability that outcome $k$ is obtained, and write $\sigma^{(k)}$ for the state on $ABCX_1\ldots X_n$ in this event. Then we can express the overall state of the system as:
\begin{equation}\label{sigmastructure}
\sigma=\sum_{k=1}^K p_k\sigma^{(k)} \otimes \ket{k}\bra{k}_R.
\end{equation}
From \eqref{structure}, and since the subspaces $\hilb{A}^{(k)}$ are orthogonal, it is clear that the reduced state $\sigma_{ABC}^{(k)}$ is equal to the state $\sigma_{AB}^{(k)}\otimes\sigma_C^{(k)}$ as defined previously, so that the $\sigma_{AB}^{(k)}$ and $\sigma_C^{(k)}$ of \eqref{structure} are indeed the appropriate reduced states of $\sigma^{(k)}$. It is also clear that the use of $p_k$ in \eqref{sigmastructure} is consistent with that in \eqref{structure}. This implies that $\sigma_{ABC}=\sum_{k=1}^Kp_k\sigma_{AB}^{(k)}\otimes\sigma_C^{(k)}=\rho_{ABC}$.

We now write $N=\{A,B,C,X_1,\ldots,X_n\}$ and observe that $\sigma$ exhibits the following properties:
\begin{enumerate}[(i)]
\item $S(RA)_\sigma-S(A)_\sigma=:S(R|A)_\sigma=S(R|B)_\sigma=0$;
\item For $J\subseteq N$ we have $S(JR)_\sigma\geq S(J)_\sigma$; \hfill(R-monotonicity)
\item $S(R)_\sigma\geq I(AB:C)_\sigma$.
\end{enumerate}
To see (ii) notice that the structure of $\sigma$ given in equation \eqref{sigmastructure}
implies, for any $J\subseteq N$,
\begin{equation}
S(RJ)_\sigma=H(\boldvec{p})+\sum_{k=1}^K p_kS(J)_{\sigma^{(k)}}\geq S(J)_\sigma,
\end{equation}
where $H$ is the Shannon entropy, $\boldvec{p}=(p_1,\ldots,p_K)$, and the inequality follows from
\cite[Thm. 11.10]{NC00}. If $A\in J$ or $B\in J$ then the inequality becomes an
equality since the $\sigma_J^{(k)}$ are supported on orthogonal subspaces.
This proves (i). Finally, (iii) follows from
\begin{equation}
  S(ABC)_\sigma-S(AB)_\sigma=\sum_{k=1}^Kp_kS(C)_{\sigma^{(k)}}\geq S(C)_\sigma-H(\boldvec{p}),
\end{equation}
and the fact that $S(R)_\sigma=H(\boldvec{p})$.

Using these properties we follow an argument similar to that used in
\cite{Mak02} for classical entropies. Notice the following chain of inequalities:
\begin{equation}\begin{split}
S(R|X_i)_\sigma&=S(RX_i)_\sigma-S(X_i)_\sigma+S(ABX_i)_\sigma-S(ABX_i)_\sigma \\
&\leq S(RX_i)_\sigma+S(RABX_i)_\sigma-S(X_i)_\sigma-S(ABX_i)_\sigma \\
&\leq S(RAX_i)_\sigma+S(RBX_i)_\sigma-S(X_i)_\sigma-S(ABX_i)_\sigma \\
&=S(R|AX_i)_\sigma+S(R|BX_i)_\sigma+I(A:B|X_i)_\sigma \\
&\leq S(R|A)_\sigma+S(R|B)_\sigma+I(A:B|X_i)_\sigma \\
&= I(A:B|X_i)_\sigma.
\end{split}\end{equation}
In the second line we used R-monotonicity, in the third and fifth lines we used strong subadditivity and for the final equality we used property (i). This implies
\begin{equation}\begin{split}
S(X_i|R)_\sigma+S(R)_\sigma &=    S(X_i)_\sigma+S(R|X_i)_\sigma \\
                            &\leq S(X_i)_\sigma+I(A:B|X_i)_\sigma.
\end{split}\end{equation}
Summing over all $i$ we obtain
\begin{equation}
  \label{dagger}
  \sum_{i=1}^n S(X_i|R)_\sigma + nS(R)_\sigma
      \leq \sum_{i=1}^n S(X_i)_\sigma + \sum_{i=1}^n I(A:B|X_i)_\sigma.
\end{equation}
Our aim here is to form inequalities for $\sigma$ which can be related back to $\rho$,
and so we need to eliminate system $R$. To this end, observe that
\begin{equation}
S(X_1\ldots X_n)_\sigma\leq S(RX_1\ldots X_n)_\sigma\leq S(RX_1\ldots X_{n-1})_\sigma+S(X_n|R)_\sigma,
\end{equation}
by R-monotonicity and SSA. Applying the second inequality here inductively we obtain
\begin{equation}
S(X_1\ldots X_n)_\sigma \leq \sum_{i=1}^n S(X_i|R)_\sigma + S(R)_\sigma,
\end{equation}
which we can substitute into \eqref{dagger} to give
\begin{equation}
S(X_1\ldots X_n)_\sigma + (n-1)S(R)_\sigma \leq \sum_{i=1}^n S(X_i)_\sigma + \sum_{i=1}^n I(A:B|X_i)_\sigma.
\end{equation}
Finally, applying property (iii) yields
\begin{equation}
S(X_1\ldots X_n)_\sigma + (n-1)I(AB:C)_\sigma
       \leq \sum_{i=1}^n S(X_i)_\sigma + \sum_{i=1}^n I(A:B|X_i)_\sigma.
\end{equation}
We have shown that the new inequalities holds for the state $\sigma$, the state after the measurement, but it remains to prove that they hold for the general state $\rho$. However, this is straightforward. Indeed, since the measurement was local at $A$ it did not alter the state on $X_1\ldots X_n$ and hence $S(X_1\ldots X_n)_\sigma=S(X_1\ldots X_n)_\rho$ and $S(X_i)_\sigma=S(X_i)_\rho$ for all $i$. Likewise, since the measurement did not affect the state on $ABC$ we must have $I(AB:C)_\sigma=I(AB:C)_\rho$. Finally, since the conditional mutual information is monotone decreasing under local maps, we must have $I(A:B|X_i)_\sigma\leq I(A:B|X_i)_\rho$. Putting all this together we obtain the result:
\begin{equation}
  S(X_1\ldots X_n)_\rho+(n-1)I(AB:C)_\rho\leq \sum_{i=1}^n S(X_i)_\rho+\sum_{i=1}^n I(A:B|X_i)_\rho,
\end{equation}
as advertised.
\end{proof}

\section{Independence of the Inequalities}\label{independence}
In the previous section we proved that certain constrained inequalities hold,
however, we have not yet justified why this result is interesting.
Let $C_n$ denote the constrained inequality of Theorem \ref{main} for a given value of $n$.
If we simplify $C_1$ then most of the terms cancel and we are left with the inequality:
\begin{equation}
I(A:B|X_1)\geq 0,
\end{equation}
which is simply one of the basic inequalities.
With this in mind, one may suspect that we have not proved anything new.

Let $n\geq 2$ be an arbitrary, but fixed, integer. We aim to show that
$C_n$ is independent of the basic inequalities and of $\{C_p\}_{p\neq n}$.
To do this we must show that $C_n$ cannot be written as a positive linear
combination of instances of $\{C_p\}_{p\neq n}$ and basic inequalities,
together with a negative linear combination of the constraints.
Let $N=\{a,b,c,x_1,\ldots,x_n\}$. Our approach will be to find a function
$g:\powset{N}\rightarrow\real$ which is submodular, monotonic, satisfies the
constraints $g(a:c|b)=g(b:c|a)=0$ and satisfies all instances of $\{C_p\}_{p\neq n}$,
but which violates $C_n$. Here we use the notation
$h(\alpha:\beta|\gamma)=-h(\gamma)+h(\alpha\cup\gamma)+h(\beta\cup\gamma)-h(\alpha\cup\beta\cup\gamma)$ for any function $h:\powset{N}\rightarrow \real$
where $\alpha,\beta,\gamma$ are disjoint subsets  of $N$.

Notice that in section \ref{result} we considered $C_p$ only as an inequality on $p+3$ parties. However, following the argument of section \ref{notation} we see that, for any $p$, there are instances of $C_p$ on $N$. This is because we can always choose $p+3$ disjoint subsets of $N$, though, of course, for $p>n$ some of these are necessarily empty.

We begin by introducing some notation. We keep $n\geq 2$ fixed, and let $N_1=\{x_1,\ldots,x_n\}$ and $N_2=\{a,b,c\}$ so that $N=N_1\dotcup N_2$. (In this section we will use lower case roman letters to represent singletons, and capital roman letters or greek letters to represent subsets of $N$). It is clear that each subset $M\subseteq N$ has a unique decomposition $M=J\dotcup K$
(we will usually write $M=JK$) with $J\subseteq N_1$ and $K\subseteq N_2$.
We can then define a function $f:\powset{N}\rightarrow \real$ by
\begin{equation}
  f(JK)=\theta_K + |J|\lambda_K -\mu_{JK},
\end{equation}
for some constants $\theta_K,\lambda_K,\mu_{JK}\in\real$, where $|J|$ is the size of the set $J$. The particular values of $\theta$ and $\lambda$ are as follows.
\begin{align}
\theta &= \left(\begin{array}{ccc}
                   & \theta_{abc} & \\
       \theta_{ab} & \theta_{ac}  & \theta_{bc} \\
          \theta_a & \theta_b     & \theta_c \\
                   &\theta_\emptyset&
                \end{array}\right)     \nonumber\\
       &:= (n+1)\left(\begin{array}{ccc}
                       & 2n^3+8n^2+4n-1 & \\
        2n^3+8n^2+4n-1 & 2n^3+5n^2+2n   & 6n^2+4n-1 \\
             2n^3+5n^2 & 4n^2+2n-1      & 3n^2+n \\
                       & 0              &
                      \end{array}\right), \\
\lambda &:= - \left(\begin{array}{ccc}
                       & 2n^3+8n^2+4n-1 & \\
        2n^3+8n^2+4n-1 & 2n^3+5n^2+2n   & 6n^2+4n-1\\
          2n^3+5n^2+2n & 4n^2+2n-1      & 4n^2+2n-1\\
                       & n^2            &
                    \end{array} \right), \label{lambdadef}
\end{align}
and the only non-zero values of $\mu_{JK}$ are
\begin{align}
  \mu_a    &= 2n^2(n+1), \\
  \mu_{ab} &= 2n(n+1)^2.
\end{align}

\begin{prop}
f is a submodular function.
\end{prop}
\begin{proof}
In \cite{Pip03} it is shown that a function $f$ is submodular if and only if all expressions of the
form $f(i:j|\alpha)$ are nonnegative, where $i,j\in N$ distinct elements,
and $\alpha\subseteq N\setminus\{i,j\}$. By considering whether $i$ and $j$ belong to
$N_1$ or $N_2$, and setting $\alpha=JK$, we arrive at three different kinds of expression:
\begin{enumerate}[(i)]
\item $f(r:s|JK)=\theta(r:s|K)+|J|\lambda(r:s|K)-\mu(r:s|JK)$;
\item $f(r:x_1|JK)=\lambda_K-\lambda_{rK}-\mu(r:x_1|JK)$;
\item $f(x_1:x_2|JK)=-\mu(x_1:x_2|JK)$;
\end{enumerate}
where $x_1,x_2\in N_1$, $r,s\in N_2$ are distinct, but otherwise arbitrary, and, for
example, $\theta(a:b|c)=-\theta_c+\theta_{ac}+\theta_{bc}-\theta_{abc}$.

By direct computation we find
\begin{align}
\theta(a:b|\emptyset) &= n(n+1)(n-2)  &\qquad \lambda(a:b|\emptyset) &= 0, \nonumber\\
\theta(a:c|\emptyset) &= n(n+1)(3n-1) &       \lambda(a:c|\emptyset) &= -(n+1)(3n-1), \nonumber\\
\theta(b:c|\emptyset) &= n(n+1)(n-1)  &       \lambda(b:c|\emptyset) &= -(n+1)(n-1), \nonumber\\
\label{theta}
\theta(a:b|c)& =n(n+1)                &       \lambda(a:b|c) &= (n+1)(n-1), \\
\theta(a:c|b)& =2n(n+1)^2             &       \lambda(a:c|b) &= -2n(n+1), \nonumber\\
\theta(b:c|a)& =2n(n+1)               &       \lambda(b:c|a) &= 0, \nonumber
\end{align}
and that the only non-zero values of $\mu(i:j|\alpha)$ are
\begin{align}
\mu(b:c|a)=\mu(b:x_1|a)=-\mu(a:b|\emptyset)                           &=2n(n+1),  \\
\mu(a:c|\emptyset)=\mu(a:x_1|\emptyset)=-\mu(c:x_1|a)=-\mu(x_1:x_2|a) &=2n^2(n+1), \\
\mu(a:c|b)=\mu(a:x_1|b)=-\mu(x_1:x_2|ab)=-\mu(c:x_1|ab)               &=2n(n+1)^2.
\end{align}
We now demonstrate in turn the non-negativity of the three different types of expressions,
(i), (ii) and (iii).
\begin{enumerate}[(i)]
\item From \eqref{theta} it follows that always $\theta(r:s|K)+|J|\lambda(r:s|K)\geq0$, so it only remains to check those cases with $\mu(r:s|JK)>0$. We find that $f(b:c|a)=f(a:c|b)=0$, and that $f(a:c|\emptyset)=n(n+1)(n-1)>0$ since $n\geq 2$.
\item From \eqref{lambdadef} we can check that always $\lambda_K-\lambda_{rK}\geq0$, so that we only have to consider cases with $\mu(r:x_1|JK)>0$. We find that $f(b:x_1|a)=(n+1)(n-1)$, $f(a:x_1|b)=0$ and $f(a:x_1|\emptyset)=2n(n+1)$.
\item Notice that all expressions of the form $\mu(x_1:x_2|JK)$ are non-positive.
\end{enumerate}

This concludes the proof.
\end{proof}

\begin{prop}
For all integer $p\geq2$, $f$ satisfies (all instances of) $C_p$, except for
$C_n$, which $f$ violates: to be precise, it violates the ``standard'' instance
on $n+3$ parties.
\end{prop}
\begin{proof}
First notice that $f$ satisfies the constraints $f(b:c|a)=f(a:c|b)=0$.
Since $C_p$ are constrained inequalities, we now fix the constrained parties,
and so consider instances of $C_p$ of the form
\begin{equation}
  L := \sum_{i=1}^pf(a:b|\alpha_i)-(p-1)f(ab:c|\emptyset)+\sum_{i=1}^pf(\alpha_i)-f(\alpha_1\ldots\alpha_p),
\end{equation}
where $\alpha_1,\ldots,\alpha_p$ are disjoint subsets of $N_1$. Observe the following
\begin{align}
f(ab:c|\emptyset) &= n(n+1)(n-1), \\
f(a:b|\emptyset)  &= n^2(n+1),    \\
f(a:b|\alpha)     &= n(n+1)(n-2) \quad\text{ for } \alpha\neq\emptyset.
\end{align}
We must now check that for all choices of $p$ and $\alpha_1,\ldots,\alpha_p$, $C_p$ is satisfied,
i.e.~$L\geq 0$, except sometimes when $p=n$.
First notice that always $\sum_{i=1}^m f(\alpha_i) - f(\alpha_1\ldots\alpha_m)=0$ and hence
\begin{equation}
  L = pn(n+1)(n-2)-(p-1)n(n+1)(n-1)+2\Delta n(n+1),
\end{equation}
where $\Delta$ denotes the number of empty sets among the $\alpha_i$.
Rearranging we obtain
\begin{equation}
  L = n(n+1)(n-p-1+2\Delta),
\end{equation}
which is negative precisely when $p>n-1+2\Delta$.
This certainly means that we must have $p\geq n$ in order for $f$ to violate $C_p$.
Let $p=n+k$. For a violation, we require that $k+1>2\Delta$, however, we must have at
least $k$ empty $\alpha_i$. Therefore, our condition becomes $k+1>2k$ and the only
violation of $C_p$ occurs when $k=0$, $p=n$ and $\Delta=0$.
\end{proof}

For any $n\geq2$ we have found a submodular function $f$ which violates $C_n$ and satisfies $C_p$ for all $p\neq n$. However, we do not know that $f$ satisfies all the basic inequalities, as it could still violate weak monotonicity. However, all the inequalities we have dealt with so far satisfy the following property: for all $i\in N$
\begin{equation}
\sum_{\omega\in\powset{N}:i\in\omega}\chi_\omega=0,
\end{equation}
where the constants $\chi_\omega$ are as defined in \eqref{ineqdef}.
Inequalities with this property are called \emph{balanced} \cite{Cha03}.
This allows us to define a new function $g:\powset{N}\rightarrow\real$ by
\begin{equation}
  g(M)=f(M)+\sum_{i\in M} c_i,
\end{equation}
for some constants $c_i\in\real$. Notice that for any balanced expression the terms involving $c_i$ will cancel, and so $g$ will take the same value as $f$. This means that for all possible values of $c_i$, $g$ will still be submodular, and will still violate $C_n$ and satisfy $C_p$ for $p\neq n$. In particular, if we choose $c_i=\max_{\alpha\subseteq\beta}\{f(\alpha)-f(\beta)\}$ for all $i$ then $g$ will be monotonic, and so certainly satisfy weak monotonicity. (Note that the $c_i$ are chosen such that if $f$ is
already monotonic, then $f=g$.) The above reasoning in particular shows that the classical
inequalities of Makarychev \emph{et al.} \cite{Mak02} are independent, a fact that was
not shown in the original paper.

\section{Variants of the Inequalities} \label{variations}
In this section we present some alternate forms of the inequalities proved in
section \ref{result}. Firstly, we can consider Theorem \ref{main} when $\rho$
is a pure state on parties $A,B,C,X_1,\ldots,X_{n+1}$. Then,
since $S(J)_\rho=S(J^c)_\rho$, we can eliminate system
$X_{n+1}$ from the inequality, and obtain the following theorem.
\begin{thm1'}
Let $\rho$ be a multipartite quantum state on parties $\{A,B,C,X_1,\ldots,X_n\}$ which satisfies the constraints:
\begin{equation} \label{constraints}
I(A:C|B)_{\rho}=I(B:C|A)_{\rho}=0.
\end{equation}
Then,
\begin{equation}\begin{split}
\sum_{i=1}^nI(A:B|X_i)_\rho &+ \sum_{i=1}^nS(X_i)_\rho+I(A:B|CX_1\ldots X_n)_\rho \\
                            &+ S(ABCX_1\ldots X_n)_\rho-S(ABC)_\rho-nI(AB:C)_\rho \geq 0.
\end{split}\end{equation}
\end{thm1'}

There is a further transformation we can make to Theorem \ref{main}.
Looking through the proof, at all stages we can allow $C$ to play the role of one of the
$X_i$. If we do this we obtain another family of constrained inequalities, and
their purified versions.
\begin{thm}
Let $\rho$ be a multipartite quantum state on parties $\{A,B,C,X_1,\ldots,X_n\}$ which satisfies \eqref{constraints}. Then:
\begin{equation}\begin{split}
\sum_{i=1}^nI(A:B|X_i)_\rho+\sum_{i=1}^nS(X_i)_\rho &+ I(A:B|C)_\rho+S(C)_\rho\\
                                                    &- S(CX_1\ldots X_n)_\rho-nI(AB:C)_\rho \geq 0.
\end{split}\end{equation}
\end{thm}

\begin{thm2'}
Let $\rho$ be a multipartite quantum state on parties $\{A,B,C,X_1,\ldots,X_n\}$ which satisfies \eqref{constraints}. Then:
\begin{equation}\begin{split}
\sum_{i=1}^nI(A:B|X_i)_\rho+\sum_{i=1}^nS & (X_i)_\rho+I(A:B|CX_1\ldots X_n)_\rho+S(ABCX_1\ldots X_n)_\rho\\
                                          &+ I(A:B|C)_\rho+S(C)_\rho-S(AB)_\rho-(n+1)I(AB:C)_\rho \geq 0.
\end{split}\end{equation}
\end{thm2'}

Previously, the only known constrained inequality for the von Neumann entropy was
found in \cite{LW05}. There it was shown that any $4$-party quantum state satisfying the
constraints $I(A:C|B)=I(B:C|A)=I(A:B|D)=0$ also satisfies the inequality
$I(C:D)\geq I(AB:C)$. Choosing $n=1$ in the forms of the theorem above, we obtain
three new four-party constrained inequalities. Since our inequalities use two of
the same constraints as \cite{LW05}, but not the other, we might expect them to be
strictly stronger, meaning we could rederive the previous result. Curiously,
however, this turns out not to be the case. Indeed, consider the function
$e:\powset{ABCD}\rightarrow\real$ given in the following table, which satisfies
the basic inequalities:

\bigskip\noindent
\begin{center}
  {\begin{tabular}{c||c|c|c|c|c|c|c|c|c|c|c|c|c|c|c|c}
  &\ $\emptyset$~&\ A~&\ B~&\ C~&\ D~&\ AB&\ AC&\ AD&\ BC&\ BD&\ CD&ABC&ABD&ACD&BCD&ABCD\\ \hline\hline
  $e$&0&5&5&2&4&6&5&5&5&5&6&6&6&5&5&4\\
\end{tabular}}
\end{center}

\bigskip\noindent
We can easily check that $e$ satisfies all the constraints of the old inequality,
however, $e(C:D|\emptyset)-e(AB:C|\emptyset)=-2<0$. Therefore, $e$ violates the old inequality.
On the other hand, each of our new four party constrained inequalities is satisfied
by $e$. Hence the old inequality is independent of the new ones.

\section{Discussion}
\label{discussion}
We have shown that an infinite family of independent inequalities hold for the von Neumann entropy. However, since these inequalities are constrained they only reveal information about the boundary of $\Sigma_n^*$, and so they are still not enough to conclude that $\overline{\Sigma_n^*}\subsetneq\Sigma_n$.

Towards this end, we note that the inequalities we proved are the same as those proved in section 3 of \cite{Mak02} for the Shannon entropy, and that our proof follows a similar outline. In \cite{Mak02} a similar family of \emph{unconstrained} inequalities for the Shannon entropy are also proved, using a method which, in some sense, generalises the constrained proof. It may be possible that this proof can be generalised to apply to the von Neumann entropy, however, it seems as though some new tools would have to be developed first.

In any case, we believe it possible that a deeper connection between classical and quantum entropy inequalities exists. We have tested many of the new non-Shannon type inequalities on quantum states, using a numerical optimisation program, and have not been able to find a single violation (although limits on processing power restrict us to Hilbert spaces with small local dimensions). We also note that all these new inequalities are balanced, and that the only entropy inequality which is known to be true in the classical but not in the quantum case, is monotonicity -- which is unbalanced.

We therefore are led to speculate whether, in fact, \emph{all} balanced
inequalities that hold for the Shannon entropy also hold for the von Neumann entropy.

\acknowledgments
We thank Beth Ruskai and Franti\v{s}ek Mat\'{u}\v{s} for discussions on information
inequalities.
The authors acknowledge support by the U.K.~EPSRC, the European Commission
(STREP project ``QCS''), the ERC (Advanced Grant ``IRQUAT''),
the Royal Society and the Philip Leverhulme Trust.

\bibliographystyle{unsrt}
\bibliography{References}

\end{document}